\documentclass[11pt]{amsart}
 \usepackage{graphicx}
\usepackage{epstopdf}
\usepackage{mydefs}
\usepackage{algorithm}
\usepackage{algorithmic}
\usepackage{fullpage}
\usepackage[foot]{amsaddr}

\usepackage{url}

\def\calT{{\mathcal{T}}}

\def\fb{\textsc{FallBack}}
\def\lp{\textsc{LowPower}}
\def\succ{\textsc{Success}}
\newcommand{\compl}[1]{\overline{#1}} 

\newcommand{\alg}[1]{\textbf{#1}}

\def\Pro{{\mathbb{P}}}
\def\Ex{{\mathbb{E}}}

\title{Towards Tight Bounds for Local Broadcasting}
\author{Magn\'us M. Halld\'orsson and Pradipta Mitra}
\address{ICE-TCS, School of Computer Science\\
 Reykjavik University\\
 101 Reykjavik, Iceland\\}
\email{mmh@ru.is, ppmitra@gmail.com}

\begin{document}
\maketitle 
\begin{abstract}
We consider the \emph{local broadcasting} problem in the SINR model,
which is a basic primitive for gathering initial information among $n$
wireless nodes.  
Assuming that nodes can measure received power, we achieve an essentially optimal constant approximate algorithm (with a $\log^2 n$ additive term). This improves
upon the previous best $O(\log n)$-approximate algorithm. Without power measurement,
our algorithm achieves $O(\log n)$-approximation, matching the previous  best result, but with a simpler approach that works under harsher conditions, such as arbitrary node failures.
We give complementary lower bounds under reasonable assumptions. 
\end{abstract}

\section{Introduction}

When a wireless ad-hoc or sensor network starts operating, 
the nodes must form an infrastructure in a distributed manner
without any information about each other.
A natural basic primitive is for each node to gather information about all other
nodes in its vicinity. The function to achieve this neighborhood
learning is called local broadcast.

In the local broadcast problem, each wireless node tries to send (the same)
message to all other nodes within a given radius.
The objective is to complete the broadcasts within the shortest
amount of time. This operation is used as a building block in
higher-layer protocols such as  routing, synchronization and
coordination. The time complexity of those protocols are often dominated by
the complexity of the local broadcast operation.

The model we use is the unstructured radio model,
which avoids any assumptions of structure or synchronization.
Nodes wake up and shut-down asynchronously, meaning nodes can be switched on at
arbitrary times, including after other nodes have started operating; they can also shut down at some arbitrary point.
There is no global clock to guide the operation of the nodes.
The distribution of the nodes can be arbitrary, possibly in worst-case position.

For wireless algorithms, the model of interference is crucial. Most
work, both on centralized and distributed algorithms, assumes
a graph-based model of interference. The most common graph-based model
is the \emph{protocol} model \cite{kumar00}, where each node has a given transmission
radius within which its messages can reach and be decoded by other
nodes, and a larger interference radius within which its transmission
will disturb (and make it impossible to decode) other messages.
More recently, the \emph{physical} model, or SINR-model, which has
been most commonly used in the engineering literature, has 
received attention in algorithms research. It has been shown to be
more faithful to reality, both experimentally and theoretically \cite{MaheshwariJD08,Moscibroda2006Protocol}.
Here, interference fades
slowly with distance, and it adds up. It is neither binary,
symmetric, nor local, all of which combine to complicate analysis of
SINR algorithms.


\subsection{Our contributions}
We seek to resolve the exact complexity of the local broadcast problem,
both with upper and lower bounds.

We give a randomized distributed algorithm that achieves close to
optimal time complexity.  For a node $x$, let $N_x$ be the number of
nodes that are reachable from $x$ in the case of no interference.  
The algorithm completes the broadcast for each node $x$ within
$O(N_x \log n  + \log^2 n)$ slots, with high probability. We do not need a carrier-sense
or collision awareness mechanism for this result (i.e., nodes have no information about the state of activity in the channel except possible reception of a message).
This matches the recently achieved results by Yu \emph{et al} \cite{Yu12}.
Our algorithm is, however, simpler. It also can operate under harsher conditions
compared to \cite{Yu12} --- with asynchronous shutting-down of nodes, the algorithm in \cite{Yu12} may fail since it depends on a network of ``leaders" to coordinate transmission decisions. 

We then provide an algorithm running in optimal $O(N_x + \log^2 n)$ time.
It operates in the same harsh model as the previous one, but assumes
that nodes receive acknowledgements for free, i.e., if they manage to
broadcast to all nodes in their broadcast region, then an acknowledgment
will be returned (in fact, it is enough if this happens with some constant probability). We show that sufficient acknowledgments can be implemented if we relax the restriction of no collision
awareness. Namely, if nodes have a ``carrier-sense'' mechanism that
allows them to verify if received signal is above a certain fixed
universal threshold, that suffices to deduce that a broadcast was successful.
Previously, no better bound was known for the case of a carrier-sense mechanism.

For a lower bound, we show that the term $\log^2 n$ is
necessary, under some assumptions. 
Instead of the SINR model, we prove the lower bound for the protocol model.
We also prove the lower bound, not for completely general algorithms, but for  ``input-determined'' algorithms, where the
behavior of the algorithm is a (random) function of the messages so far received. Though not completely general, this class
intuitively captures most reasonable algorithms possible for this problem.

Regarding the other term, a $\Omega(N_x)$ lower bound is immediate. 
There is evidence that no algorithm can work in time $o(N_x \log n)$, unless
nodes can receive information about the success of their transmissions,
but we do not have a formal proof of this.

Our results serve as further indication that the physical model is not
significantly more demanding computationally than the protocol mode,
at least for problems with uniformly sized neighborhoods like the
local broadcast problem.


\subsection{Related Work}

The local broadcast problem in the SINR model was introduced in
\cite{Goussevskaia2008Local}. The authors gave two randomized
distributed algorithms, both for the asynchronous unstructured radio model.
One is 
a simple Aloha-like protocol that applies in the case of ``known
competition'', i.e., when each node knows the number of nodes in its
proximity. The other, more involved, protocol holds without knowledge
of the competition (``unknown
competition''). The time complexity of the algorithms is
$O(N_x \log n)$ and $O(N_x \log^3 n)$, respectively, where $N_x$ is
the maximum number of nodes in any transmission range.

The bounds for unknown competition were improved in \cite{Yu11}
to $O(N_x \log^2 n)$, optimizing the algorithm of \cite{Goussevskaia2008Local}. 
Additionally, $O(\log n)$-approximate deterministic algorithms 
were given for a synchronized model where a
carrier-sense primitive was assumed to be available. 
Finally, an $O(\log n)$-approximate randomized algorithm without a
carrier-sense primitive was very recently proposed in \cite{Yu12}.
A lower bound of $\Omega(N_x + \log n)$ was also given.

Our first algorithm (without carrier-sense) thus matches the result of \cite{Yu12}, but has certain advantages.
The algorithm from \cite{Yu12} computes a maximal independent set as a set of leaders, which help other nodes to coordinate in an efficient manner. In contrast, our algorithms are simpler. They are variations of the original algorithm of \cite{Goussevskaia2008Local}, requiring no leader election phase. 
This approach has advantages in particularly harsh environments. Assuming nodes can shut-down arbitrarily (in which case no guarantee need be made about their success in local broadcasting), a leader based algorithm is undesirable. For example, in \cite{Yu12}, once a newly awaken node chooses a leader to attach itself to, it uses that leader for all future contention resolution purposes. This
would fail if a leader were to shut down in the meantime. No such problem afflicts our algorithms.

Local broadcasting is related to the radio broadcasting problem in more classical models \cite{Alon:1991:LBR:114567.114569,Kowalski04timeof,bary87}, to initialization and wake-up problems in wireless networks \cite{Jurdzinski02probabilisticalgorithms,mobicom} as well as coloring problems on disc graphs \cite{MoWa05,Graf94oncoloring}.

Recently, the SINR model has received considerable attention in the
algorithms community, starting with the work of Moscibroda and
Wattenhofer \cite{MoWa06}.
Constant-approximation factors are now known for
capacity problems, both with fixed power \cite{GHWW09,SODA11} and power
control \cite{KesselheimSODA11}. See the survey of \cite{GPW10}.
Distributed algorithms have been
given for dominating sets \cite{ScheidelerRS08}, scheduling \cite{KV10,icalp11},
coloring \cite{YWHLa11}, and connectivity and capacity \cite{PODC12}.

\section{Model}

The problem is informally as follows.
Given is a set $V$ of $n$ nodes in the plane. 
Each node wants to transmit a (single) message to all nodes within its
\emph{broadcast range} in the shortest amount of time. A local
broadcast operation is successful if all nodes have performed a successful
local broadcast.

We assume that nodes can wake up at any time asynchronously.  
The nodes are unaware of the network topology, which can be of
arbitrary (worst-case) layout. 
The nodes only have a crude bound on
$n$ (up to a polynomial factor). Without such a bound it is known that 
no sublinear algorithms are possible \cite{Jurdzinski02probabilisticalgorithms}.

There is no global clock or any synchronization among the nodes. 
In the analysis, we assume that time is divided into time-slots; 
this is justified by a standard trick of relating slotted
vs.\ unslotted Aloha (see \cite{Goussevskaia2008Local}).

In this paper, all nodes use the same power $P$, known as the \emph{uniform}
power scheme. 
We scale values so that $P=1$.

We adopt the SINR model of interference, a non-transmitting node
$v$ will successfully receive a message transmitted by node $u$ if,
\begin{equation}
\frac{P/d(u, v)^{\alpha}}{N + \sum_{w \in S \setminus \{u\}} P/d(w, v)^{\alpha}} \geq \beta\ ,
\label{eqn:sinr}
\end{equation}
where $N$ is the ambient \emph{noise}, $\beta$ is the required SINR level,
$\alpha > 2$ is the so-called path loss constant, 
$d(u,v)$ is the distance between two points $u$ and $v$,
and $S$ is  the set of senders transmitting simultaneously. 

For any subset $X$ of the plane, we use the notation $|X|$ to define the number of nodes in $X$.

We need the following two definitions:
\begin{defn}
The \emph{transmission region} $T_x$ is the ball of some fixed radius ($R_T$) around a node $x$ which $x$ can reach without any other node transmitting (i.e., $R_T =\frac1{(N \beta)^{1/\alpha}}$). Clearly, $N_x = |T_x|$.
\end{defn}

Since the signal quality (even without interference) becomes very poor near the boundaries of $T_x$, to achieve non-trivial results, one needs to define the broadcasting region as somewhat smaller than $T_x$:
\begin{defn}
The \emph{broadcasting region} $B_x$  is a ball of some fixed radius ($R_B$) around any node $x$,
containing all nodes to which $x$ would like to transmit.
We set $R_B = \phi R_T$ for a small constant $\phi$ ($\phi = \frac16$ suffices).
\end{defn}

We will use the notation $2 B_x$ to mean the ball of radius $2 R_B$ around $x$.
A probabilistic event is said to happen \emph{whp} (with high probability) if it happens with probability $1 - 1/n^c$, for some $c \ge 1$.

We assume that a node is not obliged to broadcast to nodes that woke
up after itself (or have shut-down before it broadcasts). 
This is consistent with the algorithm of
\cite{Goussevskaia2008Local}, even if not made explicit. No guarantees are made for nodes that ``live" for too short a period of time (i.e., the time elapsed between wake-up and shut-down is smaller than the claimed running time for the algorithm).

We now define formally the \emph{local broadcast} operation.
A node $x$ is successful in a given time slot if it transmits
a message and all nodes within $B_x$ can decode the message,
satisfying Eqn.~\ref{eqn:sinr}.
A local broadcast operation is successful when all the nodes have
become successful. The time complexity of a node $x$ is measured in
terms of the time that elapsed from waking up until the node halts the algorithm,
and is evaluated as a function of $N_x = |T_x|$.

\section{Results}
\begin{theorem}
There exists an algorithm for which the following holds \emph{whp}: each node $x$ successfully performs a local broadcast within $O(N_x \log n + \log^2 n)$ slots.
\label{mainth1}
\end{theorem}

We can improve this result to essentially optimal if we assume that the nodes can measure received power:

\begin{theorem}
Assume that in any slot, a node can measure the power received at its receiver (from all other transmitting nodes).
Then there exists an algorithm for which the following holds \emph{whp}: each node $x$ successfully performs a local broadcast within $O(N_x + \log^2 n)$ slots.
\label{mainth2}
\end{theorem}

Finally, a lower bound:
\begin{theorem}
In the protocol model, there exist instances on $n$ vertices such that 
\begin{enumerate}
\item There exists a broadcast neighborhood with a constant number of nodes.
\item No input-determined algorithm can complete local broadcast in this region in $o(\log^2 n)$ slots with high probability.
\end{enumerate}
\label{mainth3}
\end{theorem}

``Input-determined'' algorithms  are defined in Section \ref{lbsec}, where the theorem is proven. Informally these are algorithms whose behavior in a given slot is a (random) function of the messages received in previous slots.

\section{An $O(N_x \log n + \log^2 n)$ time Algorithm}
\label{firstalgo}

In this section, we will prove Theorem \ref{mainth1}.
Our algorithm is listed as Algorithm \ref{alg1} (\alg{LocalBroadcast1}). The  symbols $\gamma, \lambda$ used in the listing are appropriate constants.

\begin{algorithm}
\caption{LocalBroadcast1 (For any node $y$)}   
\label{alg1}                           
\begin{algorithmic}[1]                    
     \STATE $tp_y \leftarrow 0$
      \STATE $p_y \leftarrow \frac1{4n}$
    \LOOP
      \STATE $p_y \leftarrow \max\{\frac1{128 n}, \frac{p_y}{32}\}$ \label{resetline}
      \STATE $rc_y \leftarrow 0$
     \LOOP
      \STATE $p_y \leftarrow \min\{\frac1{16}, 2 p_y\}$ \label{probincrease}
     \FOR{$j \leftarrow{} 1, 2, \ldots \delta \log n$} \label{innerloop}
       \STATE $s \leftarrow 1$ with probability $p_y$ \label{choosetransmit}
       \IF{s = 1}
        \STATE transmit
       \ENDIF
       \STATE $tp_y \leftarrow tp_y + p_y$
       \IF{$tp_y > \gamma \log n$}  \label{haltcondition}
          \STATE halt;
       \ENDIF
      \IF{message received}
        \STATE $rc_y \leftarrow rc_y + 1$
        \IF{$rc_y > \log n$}
          \STATE goto line \ref{resetline} \label{resettrigger}
        \ENDIF
       \ENDIF
      \ENDFOR
     \ENDLOOP
    \ENDLOOP
\end{algorithmic}
\label{alg1fig}
\end{algorithm}

The intuition behind the algorithm is as follows. The ``right'' probability for $x$ to transmit at is about $\frac1{N_x}$ (too high, and collisions are inevitable; too low, nothing happens). The algorithm starts from a low probability, continuously increasing it, but once it starts receiving messages from others, it uses that as an indication that the ``right''  transmission probability has been reached.

To prove Thm.~\ref{mainth1}, we will first need the following definition.
\begin{defn}
For any node $x$, the event {\lp} occurs at a time slot if the received power at $x$ from other nodes, $P_x \leq \frac{1}{(4 (\beta + 4) R_B)^{\alpha}}$.
\end{defn}

\noindent The following technical Lemma follows from geometric arguments (see Appendix \ref{sec:missing} for the proof).
\begin{lemma}
If $x$ transmits and {\lp} occurs at $x$, all nodes in $2 B_x$ receive the message from $x$ (thus a successful local broadcast occurs for $x$).
\label{lowpmeanssuccess}
\end{lemma}

\noindent We will also need the following definition:
\begin{defn}
A {\fb} event is said to occur for node $y$
if line \ref{resettrigger} is executed for $y$.
\end{defn}

We will refer to the transmission probability $p_y$ for a node $y$ at 
given time slots. This will always refer to the value of $p_y$ in line \ref{choosetransmit}. We first prove a Lemma that bounds the transmission probability in any
broadcast region at a given time.
\begin{lemma}
Consider any node $x$. Then during any time slot $t \leq 10 n^2$, 
\begin{equation}
\sum_{y \in B_x} p_y \leq \frac12
\label{boundedprob}
\end{equation}
with probability at least  $1 - \frac1{n^4}$.
\end{lemma}
\begin{proof}
For contradiction, we will upper bound the probability that Eqn.~\ref{boundedprob} is violated for the first time at any given time $t$, after which we will union bound over all $t \leq 10 n^2$.

Let $\calT$ be the interval (time period) $\{t - \delta \log n + 1 \ldots t -1\}$. Then we claim,

\begin{claim}
In each time slot in the period $\calT$,
\begin{equation}
  \frac12 \geq \sum_{y \in B_x} p_y \geq \frac14 \label{probbounds}
\end{equation}
\end{claim}
\begin{proof}
The first inequality is by the assumption that $t$ is the first slot when Eqn.~\ref{boundedprob} is violated. The second is because probabilities (at most) double once every $\delta \log n$ slots (by the description of the algorithm).
\end{proof}

We now show that Eqn.~\ref{probbounds} is not possible. To that end, we show that in the $\delta \log n$ interval preceding $t$, a {\fb} will occur with high probability:

\begin{claim}
With probability $1 - \frac1{n^8}$, each node $z \in B_x$ will {\fb} once in the period $\calT$.
\label{fallbackoccurs}
\end{claim}

\begin{proof}
Fix any $z \in B_x$.
By the algorithm 
\begin{equation}
p_z \leq \frac1{16} \label{yupperbound}  
\end{equation}

Thus, at any time slot,
\begin{equation}
\Pro(z  \text{ does not transmit}) \geq \frac{15}{16}
\label{ydoesnttransmit}
\end{equation}

Now, combining Eqn.~\ref{yupperbound}  and Eqn.~\ref{probbounds}, and defining
$B = B_x \setminus \{z\}$,
\begin{equation}
  \sum_{y \in B} p_y \geq \frac3{16}
  \label{highprobothers}
\end{equation}

For $y \in B_x$ define $\succ_y$ to be the event that $y$ transmits and {\lp} occurs for $y$. By Lemma \ref{lowpmeanssuccess}, $\succ_y$ implies that $z$ will receive the message from $y$.
Thus, the probability of $z$ receiving a message from some node in $B$ in a given round is at least $\frac{15}{16}\Pro(\bigcup\limits_{y\in B}\succ_y)$.

We claim that for any $y \neq w$ (both in $B)$, the events $\succ_y$ and $\succ_w$ are disjoint. This is implicit in Lemma \ref{lowpmeanssuccess}, since $\succ_y$ means that $w$ cannot be transmitting and vice-versa. Thus, the probability of $z$ receiving a message from some node in $B$ is at least:

\begin{align*}
\lefteqn{\frac{15}{16} \Pro(\bigcup\limits_{y\in B}\succ_y)  = \frac{15}{16} \sum\limits_{y\in B} \Pro(\succ_y)} \\
& \geq \frac{15}{16} \sum\limits_{y\in B} p_y \frac12 \left(\frac14\right)^{\frac12 O\left(\frac1{\phi^2}\right)} 
\geq \frac{15}{32} \left(\frac14\right)^{\frac12 O\left(\frac1{\phi^2}\right)} \frac3{16}\ , 
\end{align*}
where we use Lemma \ref{lphappens} for the first inequality (stated and proved in Appendix \ref{sec:missing}) and Eqn.~\ref{highprobothers} for the last.

Setting $\delta \geq \frac{10}{\frac{15}{32} \left(\frac14\right)^{\frac12 O(\frac1{\phi^2})} \frac3{16}}$ and using the Chernoff bound, we can show that $z$ will receive $> \log n$ messages in $\calT$ with probability $1 - \frac1{n^8}$, thus triggering the {\fb}.
\end{proof}

Now we show that the above claim implies that Eqn.~\ref{probbounds} is not possible.
\begin{claim}
There exists a time slot in $\calT$ such that \\$\sum_{y \in B_x} p_y < \frac14$.
\end{claim}
\begin{proof}
For any $y \in B_x$, let $p_y^1$ be the value of $p_y$ in the first slot of $\calT$. Let $p_y^f$ be the value of $p_y$ in the slot when {\fb} happened for $y$. Since probabilities can at most double during $\calT$,
\begin{equation}
\sum_{y \in B_x} p_y^f \leq 2 \sum_{y \in B_x} p_y^1 \leq 1\ ,
\label{pfubound}
\end{equation}
the last inequality using the fact that $\sum_{y \in B_x} p_y^1 \leq \frac12$ (Eqn.~\ref{probbounds}).

Now by lines \ref{resetline} and \ref{probincrease} of the algorithm,
in the slot after {\fb}, $p_y = \max\{\frac{1}{128n},
\frac{p_y^f}{32}\} \leq \frac{1}{128n} + \frac{p_y^f}{32}$. Since
probabilities at most double during $\calT$, the value of $p_y$ at the
final slot of $\calT$ is at most $\frac{1}{64n} +
\frac{p_y^f}{16}$. Summing over all $y$, during the final slot of
$\calT$,

\begin{align*}
\sum_{y \in B_x} p_y  \leq \frac{n}{32n} + \sum_{y \in B_x} \frac{p_y^f}{8}
\leq \frac1{32} + \frac18 < \frac14
\end{align*}
contradicting Eqn.~\ref{probbounds}. We used Eqn.~\ref{pfubound} in the second
inequality.
\end{proof}

The proof of the Lemma is completed by union bounding over time slots 
$t \leq 10 n^2$.
\end{proof}

\noindent Now we prove that nodes stop running the algorithm by a certain time.
\begin{lemma}
Each node $x$ stops executing within $O(N_x \log n + \log^2 n)$ slots, whp.
\label{perf1}
\end{lemma}
\begin{proof}
Fix $x$. We derive four claims that together imply the lemma.

First, by the halting condition
of line \ref{haltcondition}:

\begin{claim}
The number of slots for which $p_x \geq \frac1{32}$ is $O(\log n)$.
\label{pxusuallysmall}
\end{claim}

Assume that $x$ experienced
$k$ {\fb s}.
Consider the times $t_x(1), t_x(2) \ldots t_x(k)$ when
a {\fb}  happened for $x$. Now,
\begin{claim}
$t_x(1) = O(\log^2 n)$. Also, there are $O(\log^2 n)$ slots after $t_x(k)$.
\end{claim}
\begin{proof}
The two claims are very similar. Let us prove the latter one.
Since {\fb} does not occur after $t_x(k)$, the probability doubles every $\delta \log n$ slots. Since the minimum probability is $\Omega(\frac1n)$, by $O(\log^2 n)$ slots, the probability will reach $\frac1{32}$. Once this happens, the algorithm terminates in $O(\log n)$ additional slots, by Claim \ref{pxusuallysmall}.
\end{proof}

Given the above claim it suffices 
to bound $t_x(k) - t_x(1)$. By Claim \ref{pxusuallysmall} we can also restrict ourselves to slots for which $p_x < \frac1{32}$. For these slots,
line $\ref{probincrease}$ does not need the $\min$ clause, i.e., 
$p_y \leftarrow 2 p_y$ each time line $\ref{probincrease}$ is executed. 

Define $b_i$ such that $p_x = \frac1{2^{b_i}}$ at time $t_x(i)$. 
Note that if $n$ is a power of $2$, $b_i$ is always an integer (the case of other values of $n$ can be easily managed).

We can characterize the running time between two {\fb s} as follows.
\begin{claim}
$t_x(i+1) - t_x(i) \leq (b_{i} - b_{i+1} + 5) \delta \log n$, for all $i = 1, 2 \ldots k-1$.
\label{tintermsofb}
\end{claim}
\begin{proof}
During slots in $[t_x(i), t_x(i+1))$, $p_x$ doubles every $\delta \log n$ slots (by the description of the algorithm and the fact that 
$p_x < \frac1{32}$). 
Let $b$ be  such that $p_x = \frac1{2^{b}}$ at time $t_x(i+1) - 1$. Then, 
\begin{align*}
& \frac1{  2^{b}} = \frac{2^{\left\lfloor{\frac{t_x(i+1) - t_x(i)}{\delta \log n}}\right\rfloor}}{2^{b_i}} \\
\Rightarrow\,\,  & b_i - b = \left\lfloor{\frac{t_x(i+1) - t_x(i)}{\delta \log n}}\right\rfloor
\end{align*}
By lines \ref{probincrease} and \ref{resetline} of the algorithm, $b_{i+1} \leq b + 4$, and thus, 
\begin{align*}
b_i - b_{i+1} + 4 & \geq  \left\lfloor{\frac{t_x(i+1) - t_x(i)}{\delta \log n}}\right\rfloor\\ \Rightarrow
b_i - b_{i+1} + 5 & \geq  \frac{t_x(i+1) - t_x(i)}{\delta \log n} \ ,
\end{align*}
completing the proof of the Lemma.
\end{proof}

Thus, the running time $t_x(k) - t_x(1)$ can be bounded by:
\begin{align}
\lefteqn{t_x(k) - t_x(1)}  \nonumber \\ 
& = (t_x(k) - t_x(k-1)) + (t_x(k-1) - t_x(k-2)) \nonumber \\ 
& \qquad \ldots  + (t_x(2) - t_x(1)) \nonumber \\
& \leq  ((b_{k-1} - b_{k} + 5) + (b_{k-2} - b_{k-1} + 5) \nonumber \\
  & \qquad \ldots + (b_{1} - b_{2} + 5)) \delta \log n \nonumber \\
& = (b_1 - b_k + 5 k ) \delta \log n \nonumber \\
& = O(\log^2 n +  k \log n) \ ,\label{mainrunbound1}
\end{align}
where we use Claim \ref{tintermsofb}, the non-negativity of $b_k$ and the fact that $b_i = O(\log n)$ (as $p_x = \Omega(\frac1n)$).

To complete the proof of the Lemma, we need a bound on $k$:
\begin{claim}
$k = O(N_x)$. \label{kbounded}
\end{claim}
\begin{proof}
The total number of possible transmissions that $x$ could possibly hear is $O(N_x \log n)$, whp. 
This
is because each node transmits $O(\log n)$ times, whp (by Lemma \ref{transmissiontimes} in Appendix \ref{sec:missing}) and a node can only hear messages from nodes in $T_x$ (by the definition of $T_x$).
But nodes only {\fb} once for every $\log n$ messages received (by the condition immediately preceding line \ref{resettrigger}). The claim is proven. 
\end{proof}

Applying the above claim to Eqn.~\ref{mainrunbound1}, 
$t_x(k) - t_x(1) \leq O(\log^2 n + k \log n) = O(N_x \log n +\log^2 n)$, completing the argument.
\end{proof}

The final piece of the puzzle is to show that for each node, a successful local broadcast happens whp
during one of its $\Theta(\gamma \log n)$ transmissions. 

\begin{lemma}
By the time a node halts, it has successfully locally broadcast a message, whp.
\label{lem:succ-lb}
\end{lemma}
\begin{proof}
The expected number of transmission made by a node is $\gamma \log n$ (by the algorithm). By Lemmas \ref{lphappens} and \ref{lowpmeanssuccess}, during each such transmission, local broadcast succeeds with probability $\frac12 \left(\frac14\right)^{\frac12 O(\frac1{\phi^2})}$, at least. Thus, the expected number of successful local broadcasts is $\frac12 \left(\frac14\right)^{\frac12 O(\frac1{\phi^2})} \gamma \log n$. Setting $\gamma$ to a high enough constant, and using Chernoff bounds, with high probability, a successful local broadcast happens at least once.
\end{proof}

Lemmas \ref{perf1} and \ref{lem:succ-lb} together imply Thm.~\ref{mainth1}.

\section{Improved Algorithm with \\Received Power  Measurement}
We will assume for this section that nodes can measure total received power from other nodes (even when transmitting). In hardware implementations, the received
power is usually available as RSSI (Received signal strength indicator). Additionally, filtering out one's signal (thus being able to measure received power even when transmitting) is also possible in many hardware implementations.

With this primitive we are able to design an algorithm that completes local 
broadcasting in time $O(N_x + \log^2 n)$, with high probability (thus proving Thm.~\ref{mainth2}).

Our new algorithm (Algorithm \ref{alg2})
is identical to the previous one, except for 
an extra halting condition in line \ref{haltifsuccess} --- 
A node halts if {\lp} happens, which it can clearly measure with the received power measurement primitive discussed above.

\begin{algorithm}[ht]
\caption{LocalBroadcast2 (for any node $y$)}   
\label{alg2}                           
\begin{algorithmic}[1]                    
     \STATE $tp_y \leftarrow 1$
      \STATE $p_y \leftarrow \frac1{4n}$
    \LOOP
      \STATE $p_y \leftarrow \max\{\frac1{128 n}, \frac{p_y}{32}\}$ \label{resetline2}
      \STATE $rc_y \leftarrow 0$
     \LOOP
      \STATE $p_y \leftarrow \min\{\frac1{16}, 2 p_y\}$ \label{probincrease2}
     \FOR{$j \leftarrow{} 1, 2, \ldots \delta \log n$} \label{innerloop2}
       \STATE $s \leftarrow 1$ with probability $p_y$ \label{choosetransmit2}
       \IF{s = 1}
        \STATE transmit
                \IF{{\lp} occurs} \label{haltifsuccess}
         \STATE halt;
        \ENDIF

       \ENDIF
       \STATE $tp_y \leftarrow tp_y + p_y$ \label{haltcondition2}
       \IF{$tp > \gamma \log n$}
          \STATE halt;
       \ENDIF
      \IF{message received}
        \STATE $rc_y \leftarrow rc_y + 1$
        \IF{$rc_y > \log n$}
          \STATE goto line \ref{resetline} \label{resettrigger2}
        \ENDIF
       \ENDIF
      \ENDFOR
     \ENDLOOP
    \ENDLOOP
\end{algorithmic}
\label{alg2fig}
\end{algorithm}

To show why this leads to the improved bound, 
recall the proof of Lemma \ref{perf1}. In proving that Lemma, we showed in Claim \ref{kbounded} that  $k = O(N_x)$ (where $k$ is the number of {\fb}s for $x$). We will
show instead that for Algorithm \ref{alg2}:

\begin{lemma}
$k = O\left(\frac{N_x + \log n}{\log n}\right)$
\end{lemma}
\begin{proof}
As before, since we {\fb} once for every $\log n$ received messages, it suffices to show that whp, the number of transmissions from $T_x$ that a node will hear is $O(N_x + \log n)$.

By Lemma \ref{lphappens}, for any node $x$ transmitting, {\lp} occurs with a constant probability. Thus, for any given transmission, the number of unhalted nodes in $T_x$ reduces by $1$ with some constant probability $c$. For contradiction, assume nodes in $T_x$ transmit more than $10 \frac1{c} N_x + 10 \log n$ times. Using Chernoff bounds, it is easy to show that, whp, {\lp} will occur for $> N_x$ transmitting nodes, which is a contradiction (since nodes halt once they complete a {\lp} and there are only $N_x$ nodes in $T_x$).
\end{proof}

\section{Lower Bound}
\label{lbsec}
In this section, we prove Thm.~\ref{mainth3}, thus showing that the $O(\log^2 n)$ in the running time may be necessary. 
As indicated, we prove the bound in the
 protocol model of interference \cite{kumar00}. This is a widely used and simpler model 
of wireless interference. In the protocol model, there is a transmission range $R_T$ and interference range $R_I$. A transmission from $x$ to $y$ succeeds if $d(x, y) \leq R_T$ and $d(y, z) > R_I$ for all other transmitting nodes $z$.

The algorithmic result of Section \ref{firstalgo}
 applies to the protocol model as well, i.e., local broadcasting is possible in $O(N_x \log n + \log^2 n)$ time.
Here $N_x$ is the number of nodes in the ball of radius $R_T$ around $x$ (thus, the transmission and broadcast regions are identical). The analysis for the SINR model can be applied naturally to the protocol model.

Though our lower bound result does not apply directly to the SINR model, it does apply to the type of algorithm employed in the paper. The fact that we measure success by the event {\lp} is essentially equivalent to establishing an interference perimeter around nodes. Thus, our lower bound indicates that getting rid of the $O(\log^2 n)$ factor in the SINR model, if at all possible, would have to use different techniques.

We need the following two assumptions:
\begin{enumerate}
\item Nodes do not have any ``carrier sense ability'', thus only external information they can get is a message reception.
\item The algorithm is ``input-determined'', i.e., the action of the algorithm is a function of messages it has received thus far and its own random bits. We define this precisely in Definition \ref{inputdet}.
\end{enumerate}

Note that Algorithm \ref{alg1} is clearly input-determined. Though it is possible to conjure up algorithms that are not, it is difficult to imagine how such an algorithm would help. Closing this gap remains an intriguing open problem.

We start with some definitions.
For any node $x$, and any time $t$, we define a binary function:
\begin{equation*}
I(x, t) = \left\{
\begin{array}{rl}
1 & \text{ if } x \text{ successfully decoded a message},\\
0 & \text{ otherwise.}
\end{array} \right.
\end{equation*}
Assume (without loss of generality) that nodes cannot decode messages in slots where they are transmitting.

Define $I_r(x)$ to be the string containing all bits $I(x, t)$ for $t = 1 \ldots r$. Let $T(x, t)$ define whether or not node $x$ transmits at time $t$.

Now we define precisely what we mean by input-determined:
\begin{defn}
An algorithm is said to be \emph{input-determined} if for  any node $x$ and any time slot $t$,
$$ \Pro(T(x, t) = 1 | I_{t-1}(x) = B, E_x) =  \Pro(T(x, t) = 1 | I_{t-1}(x) = B)$$
for any binary string $B\in \{0,1\}^{t-1}$ and any event $E_x$ that is a function of $(\cup_{y \in [n], r \in [1, t]} T(y, r)) \setminus \{T(x, t)\}$.
This is to say: once the reception history $I_{t-1}(x)$ is known, the behavior of the algorithm does not depend on $E_x$. This implies, 
\begin{align*}
& \Pro(T (x, t) = 1|E_x) \\ = & \sum_{B\in \{0,1\}^{t-1}} \Pro(T(x, t)= & 1 | I_{t-1}(x) = B) \Pro(I_{t-1} (x) = B|E_x)
\end{align*}
\label{inputdet}
\end{defn}

Define the string $\mathbf{0}_t$ to be the string of $t$ zeroes.
Now, consider the behavior of any algorithm. For any $t$, define $p_t$ by
$$p_t = \Pro(T (x, t) = 1|N_{t-1} )$$
for any node $x$ ($x$ is arbitrary as the nodes are indistinguishable), and where $N_{t-1}$ is the event
that in each slot up to $t-1$, there were either $0$ or more than $2$ nodes transmitting in the system.

Assume a construction where each node is in every other node's interference range (but not necessarily in its transmission range). Then
the event $N_{t-1}$ implies that $I_{t-1}(x) = \mathbf{0}_{t-1}$ for all $x$. Our lower bound will present such a construction, and then show that $N_t$ occurs with significant probability for $t = o(\log^2 n)$. Clearly, this means that no messages were decoded, thus local broadcast has not happened for any node. We now claim that the probabilities in a single slot are independent across nodes, or,

\begin{lemma}
Consider an input-determined algorithm and a node $x$. Let $\tilde T(x, t)$ define any arbitrary collection of transmissions at time $t$ by some other nodes in the system. 
Then, $$\Pro(T (x, t)|N_{t-1}, \tilde T (x, t)) = p_t$$
\label{xindependent}
\end{lemma}

\begin{proof}
First note that the event $N_{t-1} \cup \tilde T (x, t)$ meets the conditions of $E_x$ 
as laid down in Defn.~\ref{inputdet}, since it is a function of the past events, plus, events in the present excluding $T(x, t)$.

Since the algorithm is input-determined,
\begin{align*}
 P(T (x, t)|N_{t-1}, \tilde T(x, t)) \\
 = \sum_{B \in \{0, 1\}^{t-1}} 
   & \Pro(T(x, t) = 1 | I_{t-1}(x) = B) \\
   & \cdot  \Pro(I_{t-1} (x) = B|N_{t-1}, \tilde T(x, t))
\end{align*}
Note that clearly $\Pro(I_{t-1} (x) = B|N_{t-1} , \tilde T (x, t)) = 1$ for $B = \mathbf{0}_{t-1}$ and $\Pro(I_{t-1} (x)  = B|N_{t-1} , \tilde T (x, t)) = 0$ for all other $B$. Thus $P(T (x, t)|N_{t-1}, \tilde T(x, t)) = \Pro(T(x, t) = 1 | I_{t-1}(x) = \mathbf{0}_{t-1})$. A similar argument shows that  $p_t = \Pro(T(x, t) = 1 | I_{t-1}(x) = \mathbf{0}_{t-1})$, completing the
proof of the Lemma.
\end{proof}

We now provide a construction of nodes leading to the lower bound.
Consider two transmission regions that are non-overlapping, yet
close enough that they are included in each other's interference region. One will have a
constant number of nodes, the other will have $\Delta$ nodes, a value which will be set later.

\begin{figure}[ht]
\begin{center}
\includegraphics[width=4.5cm]{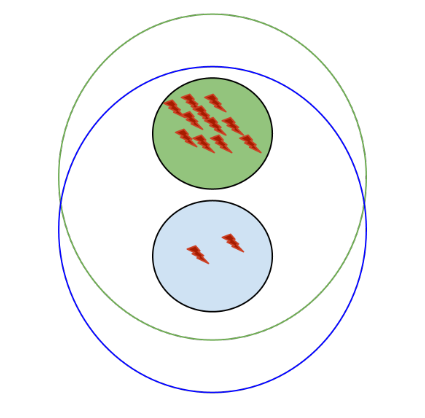}
\end{center}
\caption{The two transmission regions that are in each others interference regions. The top one has a lot of nodes, the bottom one only a few.}
\end{figure}

Partition the range $[\frac1{n^2} , 1] = \cup^r_{j=0} R_j$ where $R_0 = [-\infty, n^2)$ and for $j > 0$, $R_j = [\frac1{n^2} 16^j, \frac1{n^2} 16^{j+1})$ and $r = \Theta(\log n)$.
Consider for any $t$, the sequence $p_1, \ldots p_t$ . Define, for a range $R_i$ a weight function $w_i = \frac{|R_i \cap \{p_1 \ldots p_t \}|}{t}$.

Fix any $j$. For any $i$, define the function $f_i^j = \left(\frac{2}{e}\right)^{|i-j|+1}$.

Now we claim that (proof in Appendix \ref{sec:missing}):
\begin{lemma}
There must be range $R_j$ such that $j \leq \frac{\log n}{4}$ and 
\begin{equation}
\sum_i f_i^j w_i = O\left(\frac1{\log n}\right)  
\label{fwbound}
\end{equation}
\label{fwboundlemma}
\end{lemma}

Consider the $j$ found in the above Lemma.
Set $\Delta = \frac1{P_j}$ where $P_j = \frac4{n^2} 16^j$.

Note that by the choice of $j$, $n^2 \geq \Delta \geq n$.
Now we bound $\Pro(N_t |N_{t-1})$ for $t > 1$ (the claim also applies to $P(N_1))$.

\begin{lemma}
Assume $t$ is such that $p_t \in  R_i$ for any $i$. Then $\Pro(N_t |N_{t-1}) \geq 1 - f_i^j$.
\label{ntbound}
\end{lemma}
\begin{proof}
Follows from Lemmas \ref{ismall}, \ref{ilarge} and \ref{ieqj} (in Appendix \ref{sec:missing}).
\end{proof}

Now,
\begin{lemma}
If $t = o(\log^2 n)$, then, $\Pro(N_t) = 1 - \frac1{n^{o(1)}}$
\end{lemma}
\begin{proof}
By Lemma \ref{ntbound},
\begin{align}
\Pro(N_t) & = \Pro(N_1 )\Pro(N_2 |N_1 )\Pro(N_3 |N_2 ) \ldots \Pro(N_t |N_{t-1})\nonumber\\
= & \prod_i (1 - f_i^j )^{w_i t} 
\label{ntbound1}
\end{align}
The following claim can be proven using basic calculus:
\begin{claim}
If $x \leq \frac2{e}$, then $1 - x \geq \frac1{16^x}$
\end{claim}

Continuing with Eqn.~\ref{ntbound1} using the above claim:
 \begin{align*}
 \Pro(N_t) & = \prod_i (1 - f^i_j)^{w_i t} \geq \prod_i \frac1{16^{f_i^j w_i t}}
 = \frac1{16^{t \sum_i f^j_i w^i}} \\
 & \geq \frac1{16^{t O(1/\log n)}} =  \frac1{16^{o(\log n)}} = \frac1{n^{o(1)}}\ ,
 \end{align*}
completing the proof. The second inequality is from Eqn.~\ref{fwbound}. The equality right after that uses the assumption $t = o(\log^2 n)$.
\end{proof}

This Lemma shows that if $t = o(\log^2 n)$, then with not too small probability $\frac1{n^{o(1)}}$, the event $N_t$ occurs. Recall that $N_t$ implies that none of the nodes received any messages by time $t$, thus local broadcasting has not completed (even in the ``bottom" broadcasting range that has only a constant number of nodes). This completes the proof of Thm.~\ref{mainth3}.

\bibliographystyle{abbrv}
\bibliography{references}		

\appendix
\section{Missing Lemmas}
\label{sec:missing}

\begin{proof}{[of Lemma \ref{lowpmeanssuccess}]}
Consider any $y \in 2 B_x$. By definition of $2 B_x$, $d(x, y) \leq 2 R_B$. Now consider any other transmitting node $z$. We will show that,

\begin{claim}
$d(z, x) \leq 3 (\beta + 2) d(z, y)$
\end{claim}
\begin{proof}
By the signal propagation model, $\frac1{d(z, x)^{\alpha}}$ is the power received at $x$ from $z$. Since {\lp} occurred,
\begin{align*}
\frac1{d(z, x)^{\alpha}} & \leq \frac{1}{((4 \beta + 4) R_B)^{\alpha}}   \\
\Rightarrow\,\, d(z, x) & \geq 4 (\beta + 4) R_B
\end{align*}

By the triangle inequality, $d(z, y) \geq  d(z, x) - d(x, y)  > 
4 (\beta + 4) R_B - 2 R_B \geq 3 (\beta + 4) R_B$, proving the claim.
\end{proof}

This implies, by basic computation and summing over all transmitting $z$, that
\begin{equation}
P_y \leq \left(\frac43\right)^{\alpha} P_x
\label{relatepxpy}
\end{equation}
Now, the SINR at node $y$ (in relation to the message sent by $x$) is

\begin{align*}
  \frac{\frac1{2^{\alpha} R_B^{\alpha}}}{P_y + N} \overset{1}{\geq}
    \frac{\frac1{2^{\alpha} R_B^{\alpha}}}{\left(\frac43\right)^{\alpha}  P_x + N} 
    \overset{2}{\geq}
 \frac{\frac1{2^{\alpha} R_B^{\alpha}}}{\left(\frac43\right)^{\alpha}  
\frac{1}{((4 (\beta + 4)) R_B)^{\alpha}}  + \frac{\phi^{\alpha}}{R_B^{\alpha}\beta}} \overset{3}{\geq} \beta
\end{align*}

Explanation of numbered (in)equalities:

\begin{enumerate}
\item By Eqn.~\ref{relatepxpy}.
\item Plugging in the bound of $P_x$ (since {\lp} occurs at $x$) and noting that
$N = \frac1{\beta R_T^{\alpha}} = \frac{\phi^{\alpha}}{\beta R_B^{\alpha}}$, from the definitions of $R_T$ and $R_B$.
\item Follows from simple computation once $\phi$ is set to a small enough constant ($\phi = \frac16$ suffices).
\end{enumerate}

Thus the SINR condition is fulfilled, and $y$ receives the message from $x$.
\end{proof}

\begin{lemma}
Consider any slot $t$ and any node $z$.
Assume that in that slot, for all broadcast regions $B_x$, $\sum_{y  \in B_x} \leq \frac12$. Then, {\lp} occurs for $z$ with probability 
at least $\frac12 \left(\frac14\right)^{\frac12 O\left(\frac1{\phi^2}\right)}$.
\label{lphappens}
\end{lemma}
\begin{proof}

Let $B = B_x \setminus \{x\}$.  We first prove that there is a substantial probability that no node in $B$ transmits. Assuming this probability is $\Pro_n$

\begin{align*}
&  \Pro_n \geq \prod\limits_{w \in B} (1 - p_w)  \geq \prod\limits_{w \in B_x} (1 - p_w)  \geq \left(\frac14\right)^{\sum_w p_w} \geq \left(\frac14\right)^{\frac12}
\end{align*}
The third inequality is from Fact 3.1 \cite{Goussevskaia2008Local}, and the last from the bound $\sum_w p_w \leq \frac12$.

Let $\Pro_T$ be the probability that no other node transmits in $T_x$. Since $R_B = \phi R_T$, $T_x$ can be covered by $O(\frac1{\phi^2})$ broadcast regions (this can be shown using basic geometric arguments). 
Thus, 

\begin{equation}
\Pro_T \geq \Pro_n^{O\left(\frac1{\phi^2}\right)} \geq \left(\frac14\right)^{\frac12 O\left(\frac1{\phi^2}\right)}
\end{equation}

Since no other node in $T_x$ is transmitting, we only need to bound the signal received from outside $T_x$. 

To this end, we need the following Claim (which is a restatement of Lemma 4.1 of \cite{Goussevskaia2008Local} and can be proven by standard techniques):

\begin{claim}
Assume that for all broadcast regions $B_x$, \\$\sum_{y  \in B_x} p_y \leq \frac12$.
Consider a node $x$. Then the expected power received at node $x$ from nodes not in $T_x$ can be upper bounded by 

$$\frac18 \frac{\alpha -1}{\alpha -2} 3^3 2^{\alpha -2} \frac{\phi^2}{R_B^{\alpha}} \leq \frac{1}{2 (4 (\beta + 4) R_B)^{\alpha}}$$ for appropriately small $\phi$.
\end{claim}

Then by Markov's inequality, with probability at least $\frac12$, the power
received from nodes outside of $T_x$ is at most $\frac{1}{ (4 (\beta + 4) R_B)^{\alpha}}$.

Thus, with probability $\frac12 \Pro_T$, {\lp} occurs at $x$, proving the Lemma.
\end{proof}

\begin{lemma}
With high probability, each node transmits at least $\frac12 \gamma \log n$ times, and at most $2 \gamma \log n$ times. 
\label{transmissiontimes}
 \end{lemma}
 \begin{proof}
By the description of the algorithm, when the node stops, its total transmission probability is $\gamma \log n$.
By the standard Chernoff bound, the actual number of transmissions is very close to this number, whp.
\end{proof}

\begin{lemma}
If $p_t \in R_i$ and $i < j$, then $\Pro(N_t |N_{t-1}) > 1 - \left(\frac2{e}\right)^{j-i+1}$.
\label{ismall}
\end{lemma}
\begin{proof}
Essentially, by Markov's inequality. Let the number of nodes transmitting at time
$t$ be $T$. Then, $\Ex(T) = p_t \Delta \leq \frac1{n^2} 16^{i+1} \frac{n^2}{4 \cdot 16^j}$ (by the choice of $p_t$). Now, if $\compl{N_t}$ is the event of
$N_t$ not occurring, then

$$\Pro(\compl{N_t} | N_{t-1}) \leq \Pro(T \geq 1 | N_{t-1}) \leq \Ex(T | N_{t-1}) = \frac1{4 \cdot 16^{j-i-1}}$$
This implies the Lemma after some elementary manipulations.
\end{proof}

\begin{lemma}
Let $p_t \in R_i$ and $i > j$, then $\Pro(N_t | N_{t-1}) > 1 - \left( \frac2{e} \right)^{i-j+1}$.
\label{ilarge}
\end{lemma}
\begin{proof}
Consider the following Chernoff-type bound: Let $\{X_i\}$ be independent
Poisson trials such that $X = \sum_i X_i$ and $\mu = \Ex(X)$. Then, $\Pro(X \leq (1 - \delta)\mu) \leq e^{-\mu \delta^2/2}$ . See
Thm.~4.5 of \cite{DBLP:books/daglib/0012859} for a reference.

Now let $X$ be the number of transmissions in the slot and we would like to lower bound the
probability of there being at most $2$ of them. Note that Chernoff bound needs independence between the variables, but we have shown that in Lemma \ref{xindependent}.

Set $\delta = 1 - 2/\mu$, which means that $(1 - \delta)\mu = 2$. It is easy to verify that $\mu \geq 4$. Thus,
\begin{align*}
\Pro(\compl{N_t}|N_{t-1}) 
 & \leq \Pro(X \geq 2|N_{t-1}) \leq e^{-(1-\frac2{\mu})^2 \mu/2} \\ 
 & \leq e^{-(\mu/8)} = \exp(- \frac{16^i}{16^j})
\end{align*}
which implies the Lemma after some calculations.
\end{proof}

\begin{lemma}
If $p_t \in R_j$, $\Pro(N_t|N_{t-1}) \geq 1 - \frac2{e}$
\label{ieqj}
\end{lemma}
\begin{proof}
To see this, note that $N_t$ occurs iff the number of nodes transmitting in the slot is
not $1$. So we need to upper bound the probability of exactly $1$ node transmitting.
Note that by Lemma \ref{xindependent}, the transmission probabilities at time $t$ are iid. Let this iid
probability be $p$. The probability of exactly on transmission is $\Delta p (1-p)^{\Delta-1}$ which can be seen to be upper bounded by $\frac2{e}$ for large enough $\Delta$ using calculus.
\end{proof}

\begin{proof}{[of Lemma \ref{fwboundlemma}]}
By the definition of $w_i$,
\begin{equation}
  \sum_{i=0}^r w_i = 1 \label{wone}
\end{equation}

Now it is elementary (using
bounds for geometric series) to check that for any $i$,

\begin{equation}
\sum_j f_i^j = \Theta(1) \ .  \label{fbound}
\end{equation}

We can see that,
$$\sum_j \sum_i f_i^j w_i = \sum_i w_i \left(\sum_j f_i^j\right) = \Theta\left(\sum_i w_i\right) = \Theta(1)$$
The first equality is rearrangement, the rest follows from Eqns. \ref{fbound}
and \ref{wone}.  The Lemma now follows from noting that the sum over $j$ has $\Theta(\log n)$ terms.
\end{proof}

%
%
%

\end{document}